\documentclass[12pt]{iopart}
\usepackage{amsthm}
\usepackage{amsfonts}
\usepackage{amssymb}
\usepackage{graphicx}
\usepackage{tikz}
\usetikzlibrary{decorations.markings}
\usetikzlibrary{arrows}
\usepackage{color}
\usepackage{hyperref}

\newtheorem{theorem}{Theorem}

\newtheorem{lemma}{Lemma}

\newtheorem{corollary}{Corollary}

\newcommand*\cyl[1]{\mathrm{cyl}( #1 )}

\newcommand*\Emod[3]{#1 \equiv #2 ~ (\mathrm{mod~} #3)}

\newcommand*\Modop{~\mathrm{mod}~}

\begin{document}
\title{A ``problem of time'' in the multiplicative scheme for the $n$-site hopper}
\author[a,b,c]{Fay~Dowker$^{1,2,3}$, Vojt\v{e}ch~Havl\'\i\v{c}ek$^{4}$, Cyprian~Lewandowski$^{5}$ and Henry~Wilkes$^{1}$}
\address{$^1$ Theoretical Physics Group, Blackett Laboratory, Imperial College, London, SW7~2AZ, UK}
\address{$^2$ Institute for Quantum Computing, University of Waterloo, ON, N2L~2Y5, Canada}
\address{$^3$ Perimeter Institute for Theoretical Physics, 31 Caroline Street North, Waterloo, Ontario, N2L~2Y5, Canada}
\address{$^4$ Quantum Group, Department of Computer Science, University of Oxford, Wolfson Building, Parks Road, Oxford, OX1~3QD, UK}
\address{$^5$ Department of Physics, Massachusetts Institute of Technology, Cambridge, MA~02139, USA}
\ead{\mailto{hw2011@ic.ac.uk}}

\begin{abstract}
Quantum Measure Theory (QMT) is an approach to quantum mechanics, based on the path integral, in which quantum theory is conceived of as a generalised stochastic process. One of the postulates of QMT is that events with zero quantum measure do not occur, however this is not sufficient to give a full picture of the quantum world. Determining the other postulates is a work in progress and this paper investigates a proposal called the Multiplicative Scheme for QMT in which the physical world corresponds, essentially, to a set of histories from the path integral. This scheme is applied to Sorkin's $n$-site hopper, a discrete, unitary model of a single particle on a ring of $n$ sites, motivated by free Schr\"{o}dinger propagation. It is shown that the multiplicative scheme's global features lead to the conclusion that no non-trivial, time-finite event can occur.
\end{abstract}


\section{Introduction} \label{Introduction}

One motivation for reformulating quantum mechanics beyond the Copenhagen interpretation is to understand quantum mechanics without observers or measurements. Another is to bring quantum theory into harmony with relativity and the four dimensional world view of General Relativity. Such a framework would be relevant for quantum cosmology and could contribute to the construction of a theory of quantum gravity. A promising starting point for a realist and relativistic formulation of quantum mechanics is the path integral or sum-over-histories approach to quantum mechanics, initiated by P.A.M. Dirac \cite{1933_dirac_action_principle_qm_original_Paul_Dirac} and developed by R.P. Feynman \cite{1948_feynman_spacetime_qm_Richard_Feynman}. In the approach, for a given system, a \textit{history} is the most complete description possible of the system throughout spacetime. For a single particle on a fixed spacetime a history   would be a spacetime path, for two particles it would be a pair of spacetime paths and for a quantum field in a fixed spacetime it would be a field configuration on that spacetime. The path integral shifts the focus from the state vectors, Hilbert space, Schr\"{o}dinger evolution, and state vector collapse of canonical quantum mechanics to spacetime histories,  spacetime events and their amplitudes. This allows us to consider statements about the quantum world without appealing to the notion of collapse and measurement, and without having to perform the split of the universe into quantum system and classical measuring system necessary in Copenhagen quantum mechanics.

The path integral approach to quantum foundations has been championed particularly by J.B. Hartle and R.D. Sorkin and there is much common ground between Hartle's path integral version of Decoherent Histories \cite{1989_Quantum_Cosmology_James_Hartle, 1993_spacetime_approach_James_Hartle, 1993_spacetime_QM_lectures_Jame_Hartle} and Sorkin's Quantum Measure Theory (QMT) which conceives of quantum theory as a generalised stochastic process \cite{1987_role_of_time_Rafael_Sorkin, 1991_problems_with_causailty_Rafael_Sorkin, 1991_sum_over_histories_EPRB_Sukanya_Sinha_and_Rafael_Sorkin, 1994_sum_rules_Rafael_Sorkin}. Quantum physics, in the path integral approach, is rooted in spacetime. The dynamics and initial condition of a quantum system are encoded in a Schwinger-Keldysh double path integral for the \textit{quantum measure} \cite{1948_feynman_spacetime_qm_Richard_Feynman, 1987_role_of_time_Rafael_Sorkin}, or, equivalently, \textit{decoherence  functional}  \cite{1989_Quantum_Cosmology_James_Hartle,1994_sum_rules_Rafael_Sorkin, 1986_measurements_distributed_in_time_Carlton_Caves} on spacetime events. Some frontiers of current knowledge in the path integral  approach to quantum theory are signposted by more-or-less technical questions such as how to define the path integral for quantum mechanics as a mathematically well-defined integral over path space, including both ``ultraviolet'' problems associated with the continuity of spacetime (see \textit{e.g.} \cite{2012_towards_a_fundamental_Rafael_Sorkin, 2010_extending_quantum_measure_Fay_Dowker_and_Steven_Johnston_and_Sumati_Surya, unpublished_Path_Integral_Robert_Geroch})  and ``infrared''  problems associated with questions involving arbitrarily long times \cite{2012_towards_a_fundamental_Rafael_Sorkin, 2011_2_site_quantum_random_walk_Stan_Gudder_and_Rafael_Sorkin}.  The central \textit{conceptual} question  in QMT is: ``What, in the path integral approach to quantum theory, corresponds to the physical world?''

In  Quantum Measure Theory, the concept of a \textit{co-event} has been proposed as that which represents the physical world. A co-event is an answer to every yes-no physical question that can be asked about the world, once the class of  spacetime histories has been fixed. In this paper we investigate one proposal -- the \textit{multiplicative scheme} -- for the physical laws governing co-events. We work in the context of a one parameter family of discrete unitary quantum models, the $n$-site hopper  \cite{2012_towards_a_fundamental_Rafael_Sorkin}. The discreteness can be viewed as a ploy to sidestep the ultraviolet problems, referred to above, that arise in trying to define the path integral in the continuum, or more positively as something we may actually want to keep in the end. Indeed, Sorkin's motivation for developing  Quantum Measure Theory is to employ it to build a quantum theory of gravity in which continuum spacetime emerges from a fundamentally discrete substructure.

In section~\ref{n-site} we introduce the $n$-site hopper as a quantum system within Quantum Measure Theory. We define the decoherence functional, the quantum measure and the law of preclusion. In section~\ref{Nothing_Happens} we define the concept of a co-event and the Multiplicative Scheme. We prove the main result of the paper, that every non-trivial finite time event is a subevent of an event of measure zero. In the Multiplicative Scheme this implies that no non-trivial finite time event happens. We discuss the implications of this result in section~\ref{analysis}.
  
\section{The $n$-site Hopper} \label{n-site}

The $n$-site hopper model, proposed by Sorkin in \cite{2012_towards_a_fundamental_Rafael_Sorkin}, is a unitary quantum model of a particle on a ring of $n>1$ spatial sites $i \in \mathbb{Z}_n$, with discrete times steps $t=0,1,2\dots$. For more details on the 2-site hopper see \cite{2010_extending_quantum_measure_Fay_Dowker_and_Steven_Johnston_and_Sumati_Surya, 2011_2_site_quantum_random_walk_Stan_Gudder_and_Rafael_Sorkin} and on the 3-site hopper see \cite{2013_3_site_energy_Rafael_Sorkin}. 
  
  In QMT \cite{2007_co-events_Rafael_Sorkin, 2007_anhomomorphic_logic_Rafael_Sorkin, 2012_logic_to_quantum_Rafael_Sorkin} a system is defined by a triple $(\Omega, \mathfrak{A}, D)$ where $\Omega$ is the \textit{set of histories}, $\mathfrak{A}$ is the \textit{event algebra} and $D$ is the \textit{decoherence functional}. An event is a set of histories and the event algebra is the set of all events to which the theory assigns a measure. The event algebra is, then, a subset of the power set of $\Omega$. The decoherence functional is a function with two arguments $D: \mathfrak{A}\times\mathfrak{A} \rightarrow \mathbb{C}$ such that 
\begin{itemize}
	\item $D(A,B) = D(B,A)^* \quad \forall ~ A,B \in \mathfrak{A} $ \,;
	\item For any finite collection of events $A_1, \ldots, A_m \in \mathfrak{A} $, the $m\times m$ matrix \\$M_{ab}:=D(A_a,A_b)$ is positive semi-definite\,;
	\item $D(\Omega, \Omega ) = 1$\,;
	\item $D(A \cup B,C)= D(A,C) + D(B,C) \quad \forall ~ A,B,C\in \mathfrak{A} $ such that $A \cap B = \emptyset$\,.
\end{itemize}
The \textit{quantum measure}, $\mu(E)$, of an event $E\in \mathfrak{A}$ is given by the diagonal of the decoherence functional $D$:
\begin{equation}
	\mu(E):= D(E,E)\,.
\end{equation}
  
We will now define each of these entities for the $n$-site hopper. First, the spacetime lattice is fixed so a history for the system is simply a path, $\gamma$, the hopper can take on the lattice. Each history is an infinite sequence of spatial sites which can be conceived of as a function $\gamma: \mathbb{N} \rightarrow\mathbb{Z}_n$ or as an infinite string of sites. For example, here is the beginning of a particular history displayed in the two ways:
\begin{equation}
	 \gamma(t) =\cases{0 & for $t=0$
	 	\\ 3 & for $t=1$
	 	\\ 5 & for $t=2$
	 	\\ 5 & for $t=3$
	 	\\ 2 & for $t=4$
		\\ \, \vdots & }
	 	\quad \leftrightarrow \quad \gamma = 03552\ldots \, .
\end{equation}
Figure \ref{n_site_diagram} gives a visualisation of this history for the $8$-site hopper. The history space, $\Omega(n)$, is the set of all possible paths $\gamma$.

\begin{figure}[t!]
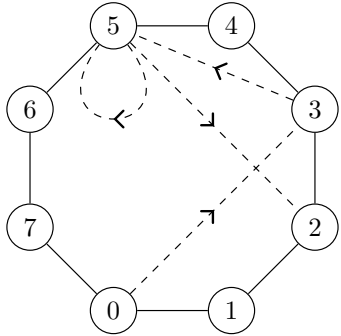

	\newlength{\sidelength}
	\setlength{\sidelength}{0.1\textwidth}
	\newlength{\sideheight}
	\setlength{\sideheight}{0.707106781\sidelength}
	\newcommand*\addsite[2]{\node[circle,draw,line width=0.01\sidelength] (site_#2) at (#1) {$#2$} }
	\newcommand*\pathbetween[2]{\draw[line width=0.01\sidelength,dashed,postaction={decorate}] (site_#1) -- (site_#2)}
	\begin{indented}
	\item[] \tikz[decoration={markings,mark=at position 0.5 with {\arrow[line width=0.02\sidelength]{angle 90}}}]{
			\addsite{0,0}{0};
			\addsite{\sidelength,0}{1};
			\addsite{\sidelength + \sideheight,\sideheight}{2};
			\addsite{\sidelength + \sideheight,\sidelength + \sideheight}{3};
			\addsite{\sidelength,\sidelength + 2\sideheight}{4};
			\addsite{0,\sidelength + 2\sideheight}{5};
			\addsite{-\sideheight, \sidelength + \sideheight}{6};
			\addsite{-\sideheight, \sideheight}{7};
			\draw[line width=0.01\sidelength] (site_0) -- (site_1) -- (site_2) -- (site_3) -- (site_4) -- (site_5) -- (site_6) -- (site_7) -- (site_0) ;
			\pathbetween{0}{3};
			\pathbetween{3}{5};
			\pathbetween{5}{2};
			\draw[line width=0.01\sidelength,dashed,postaction={decorate}] (site_5) .. controls +(0.75\sidelength,-\sidelength) and +(-0.75\sidelength,-\sidelength) .. (site_5) ;
		}
	\end{indented}
	\caption{\label{n_site_diagram}A visualisation of the $8$-site hopper. The dashed lines represent the path $03552\ldots$ up to time $t=4$. Note that the shape of the dashed lines is not physically meaningful.}
\end{figure}

Each path $\gamma$ that starts at site $i$ at $t=0$ will have an initial amplitude of $\psi_i$ and $\sum_{i = 0}^{n-1} |\psi_i|^2 = 1$.  In the canonical approach the initial amplitudes are the initial state, $\vert \psi \rangle = \sum_i \psi_i \vert i \rangle$. However, in a sum over histories approach to quantum mechanics where the focus is on the paths, not the state vector, these initial amplitudes should be thought of as summarising as much of the past as is necessary to make further predictions. From then on, at each time step the hopper transfers from site $j$ to $k$ with amplitude
\begin{equation}
	U(n)_{jk} := \frac{{\omega_n}^{(j-k)^2}}{\sqrt{n}} \,,
\end{equation}
where
\begin{equation}
	{\omega_n}:= \cases{
		e\left( \frac{1}{2n} \right) & for even $n$
		\\ e\left({\frac{1}{n}} \right) & for odd $n$
	}
\end{equation}
and $ e(x) := \rme^{\rmi 2\pi x}$. So for $n$ odd, $\omega_n$ is an $n^{\mathrm{th}}$ root of unity and for $n$ even, $\omega_n$ is a $(2n)^{\mathrm{th}}$ root of unity. The form of $U(n)$ is motivated by the propagator of a nonrelativistic free particle of mass $m$ in continuum $1$-$d$ space from $x_i$ to $x_f$ in time $T$, given by \cite{1948_feynman_spacetime_qm_Richard_Feynman}
\begin{equation}
	K^{(\mathrm{free})}(x_f,T;x_i,0) = \sqrt{\frac{m}{\rmi 2\pi \hbar T}} \, \mathrm{exp}\left( \frac{\rmi m(x_f-x_i)^2}{2\hbar T}\right)\,.
\end{equation}
Moreover, $U(n)$ is a unitary $n \times n$ matrix,  $U(n)_{jk}$ only depends on the distance between sites rather than their absolute position, and $U(n)_{j(k \pm n)} = U(n)_{jk}$ reflects the periodicity of the ring.

The form of the model as a  process gives the space of histories a temporal structure. For a given time $t\in \mathbb{N}$ consider a finite path consisting of the first $t+1$ sites of a history. We call such a truncated history a $t$\textit{-path} $\gamma_t$. The set of all such $t$-paths, for a given $t$, is $\Omega_t(n)$ with  cardinality  $|\Omega_t(n)| = n^{t+1}$.

Each  $t$-path, $\gamma_t$, is associated to an event called a cylinder set which corresponds to the physical statement  ``the hopper's history agrees with $\gamma_t$ for the first $t+1$ sites.'' The cylinder set, $\cyl{\gamma_t}$,  is formally defined by
\begin{equation}
	\cyl{\gamma_t}:=\left\lbrace \gamma \in \Omega(n)~ \middle\vert ~ \gamma(t')=\gamma_t(t') \mathrm{~for~} 0 \leq t' \leq t \right\rbrace\,.
\end{equation}

For any two cylinder sets, either they are disjoint or one  contains the other. The cylinder set for $\gamma_t$ can be expressed as a disjoint union of cylinder sets of $t'$-paths where $t'>t$ in the following way. First, for the $t$-path $\gamma_t$, consider the $(t+1)$-path which is the \textit{extension} of $\gamma_t$, by the site $j  \in \mathbb{Z}_n$ at time $t+1$. We write this extension as $\gamma_tj$: 
\begin{equation}
	\gamma_tj (t'):= \cases{
		\gamma_t(t') & for $0\leq t' \leq t$ 
		\\ j & for $t'=t+1$\,. }
\end{equation}
The cylinder set of $\gamma_t$ is the disjoint union of all such extensions of $\gamma_t$:
\begin{equation}
	\cyl{\gamma_t} = \bigcup \limits_{j=0}^{n-1} \cyl{\gamma_tj} \label{cylinder_expansion}\,.
\end{equation}
Each cylinder set in this union can be similarly expressed as a disjoint union of cylinder sets of extended paths, 
and so on. 

We will call an event $E$ which can be expressed as a finite union of cylinder sets \textit{time-finite} because we can determine whether any given history $\gamma$ is an element of $E$ or not by only looking at the first $(t+1)$ values of $\gamma$ for some finite $t$. For a time-finite event, $E$, and each $t\in \mathbb{N}$ we define a corresponding set of $t$-paths
\begin{equation}
	E_t := \left\lbrace \gamma_t \in \Omega_t(n) ~ \middle\vert ~ \cyl{\gamma_t} \subseteq E \right\rbrace\,.
\end{equation}
we also define the \textit{defining time} $t_E$, which is the \textit{least} time $t$ such that
\begin{equation}
	E = \bigcup \limits_{\gamma_t \in E_t} \cyl{\gamma_t}\,.
\end{equation}

For example, consider the event $E$ corresponding to the statement ``The hopper is at site $0$ at $t=0$ or at site $1$ at $t=1$'' for the $2$-site hopper:
\begin{equation}
E = \cyl{0} \cup \cyl{11} = \cyl{00} \cup \cyl{01} \cup \cyl{11}\,.
\end{equation} 
Then, we have $E_0=\lbrace 0 \rbrace$ and $E_1=\lbrace 00, 01, 11 \rbrace$ and the defining time $t_E=1$.
 
The \textit{path amplitude}  $a[\gamma_t]$ of a $t$-path is given by the product of the sequence of transfer amplitudes and the initial amplitude
\begin{equation}
	a[\gamma_t]:= \psi_{\gamma_t(0)} \prod \limits_{t'=0}^{t-1} U(n)_{\gamma_t(t')\,\gamma_t(t'+1)}\,.
\end{equation}  
The decoherence functional of two $t$-path cylinder sets is given by
\begin{equation}
	D \left(\cyl{\gamma_t},\cyl{\gamma_t'} \right) = a[\gamma_t] a[\gamma_t']^* \delta_{\gamma_t(t),\gamma_t'(t)}\,,
\end{equation}
where the Kronecker delta ensures the paths meet at their end points and is a feature of unitary systems \cite{1993_spacetime_QM_lectures_Jame_Hartle}. The decoherence functional $D(E, F)$ of two time-finite events is then given by 
\begin{equation}
	D(E,F) = \sum \limits_{\gamma_t \in E_t} \sum \limits_{\gamma_t' \in F_t} D\left(\cyl{\gamma_t},\cyl{\gamma_t'} \right)\,,
\end{equation}
where $t = \mathrm{max}\{t_E, t_F\}$ is the greater of the two defining times.
\begin{lemma}
	The sum
	\begin{equation}
		D(E,F) = \sum \limits_{\gamma_t \in E_t} \sum \limits_{\gamma_t' \in F_t} D\left(\cyl{\gamma_t},\cyl{\gamma_t'} \right)\,,
	\end{equation}
	is independent of $t$ if $t \ge \mathrm{max}\{t_E,t_F\}$. 
\end{lemma}
\begin{proof}
	\begin{eqnarray}
		D(E,F)& = \sum \limits_{\gamma_t \in E_t} \sum \limits_{\gamma_t' \in F_t}a[\gamma_t] a[\gamma_t']^* \delta_{\gamma_t(t),\gamma_t'(t)}
		\\ & = \sum \limits_{\gamma_t \in E_t} \sum \limits_{\gamma_t' \in F_t}a[\gamma_t] a[\gamma_t']^* \sum \limits_{f=0}^{n-1} U_{\gamma_t(t),f} {U_{\gamma_t'(t),f}}^*
		\\  & = \sum \limits_{f=0}^{n-1} \sum \limits_{\gamma_t \in E_t} a[\gamma_t] U_{\gamma_t(t),f}\sum \limits_{\gamma_t' \in F_t} a[\gamma_t']^* {U_{\gamma_t'(t),f}}^*
		\\ & =  \sum \limits_{f=0}^{n-1} \sum \limits_{\gamma_{t+1} \in E_{t+1} \vert_f} a[\gamma_{t+1}] \sum \limits_{\gamma_{t+1}' \in F_{t+1} \vert_f} a[\gamma_{t+1}']^*
		\\ & = \sum \limits_{\gamma_{t+1} \in E_{t+1}} \sum \limits_{\gamma_{t+1}' \in F_{t+1}} D\left(\cyl{\gamma_{t+1}},\cyl{\gamma_{t+1}'} \right) \, ,
	\end{eqnarray}
	where
	\begin{equation}
		E_t \vert_f := \left\lbrace \gamma_t \in E_t ~ \middle\vert ~ \gamma_t(t)=f \right\rbrace \,.
	\end{equation}
\end{proof}
This gives us the quantum measure of a time-finite event $E$ with defining time $t_E$:
\begin{eqnarray}
	\mu(E) &= \sum \limits_{\gamma_t \in E_t} \sum \limits_{\gamma_t' \in E_t} D\left(\cyl{\gamma_t},\cyl{\gamma_t'} \right)
	\\ &= \sum \limits_{f=0}^{n-1} \left\vert \sum \limits_{\gamma_t \in E_t\vert_f} a[\gamma_t] \right\vert^2 \label{lower_time_measure}\,,
\end{eqnarray}
where $f$ is the final site of the $t$-paths and $t \geq t_E$.

Since the quantum measure of time-finite events is defined, the event algebra $\mathfrak{A}(n)$ for the $n$-site hopper certainly includes all time-finite events. If $\mu$ were a classical measure it would have an extension from the semi-ring of time-finite events to the full sigma algebra generated by the time-finite events and $\mathfrak{A}$ would then be this sigma algebra. However, the standard extension theorems cannot be used for quantum measures and we do not yet know to which events the measure can be extended \cite{2012_towards_a_fundamental_Rafael_Sorkin, 2010_extending_quantum_measure_Fay_Dowker_and_Steven_Johnston_and_Sumati_Surya}. For the purposes of this paper we will only need to consider the time-finite events so we will not address this question here.

\subsection{Preclusions} \label{Preclusion}

In QMT, the quantum measure $\mu$ is viewed as a generalisation of a probability measure: quantum theories are
generalisations of classical stochastic processes. Classical theories, in this view, are special cases of 
quantum measure theories.\footnote{Classical and quantum theories are respectively the first and second levels in a countable, nested hierarchy of measure theories characterised by how much interference there is between histories  \cite{1994_sum_rules_Rafael_Sorkin}.}
In general,  quantum interference means that the quantum measure does not satisfy the Kolmogorov sum rules and cannot be interpreted as a probability measure. The question then is,  what is the physical significance of the quantum measure? The full answer is yet to be worked out, but at the current stage of development QMT assigns a central role to the concept of \textit{preclusion} \cite{1984_everett_interpretation_Robert_Geroch}: an event of zero measure is precluded (``almost surely does not happen'').  It is useful to adopt the term \textit{null event} to mean an event of measure zero and the Law of Preclusion is then:  A null event is precluded. 

Null events are therefore of central importance, so we will now investigate their structure in the case of the $n$-site hopper. For each $t$-path $\gamma_t$ we define the phase, $\tilde{a}[\gamma_t]$, of the path by
\begin{equation}
	\tilde{a}[\gamma_t] := n^{t/2} \prod \limits_{t'=0}^{t-1} U(n)_{\gamma_t(t'),\gamma_t(t'+1)}\,.
\end{equation}
Note that this definition of the phase of a $t$-path does not include the phase of the initial amplitude but only depends on the structure of the path. Then
\begin{eqnarray}
	\tilde{a}[\gamma_t] &= {\omega_n}^{(\gamma_t(0) - \gamma_t(1))^2} {\omega_n}^{(\gamma_t(1) - \gamma_t(2))^2} \ldots {\omega_n}^{(\gamma_t(t-1) - \gamma_t(t))^2}
	\\ &= {\omega_n}^p \, ,
\end{eqnarray}
where $p \in \mathbb{Z}_{n'(n)}$ and
\begin{equation}
	n'(n):= \cases{
		2n & for even $n$
		\\ n & for odd $n$.
	} 
\end{equation}
Now, for each time-finite event $E \in \mathfrak{A}$, each $i,f\in \mathbb{Z}_n$, $p \in \mathbb{Z}_{n'(n)}$, and each $t\ge t_E$ let us define
\begin{equation}
	s^{t}_{ifp} (E) := \left\vert \left\lbrace \gamma_t \in E_t ~ \middle\vert ~ \gamma_t(0)=i, \, \gamma_t(t)=f, \, \tilde{a}[\gamma_t] = {\omega_n}^p \right\rbrace \right\vert \, ,
\end{equation}
where $\vert \cdot \vert$  is the cardinality of a set. In other words, $s^{t}_{ifp} (E)$ is the number of $t$-paths in $E_t$ which begin at $i$, end at $f$ and have phase ${\omega_n}^p$.

Recall that the quantum measure is
\begin{equation}
	\mu(E) = \sum \limits_{f=0}^{n-1} \left\vert \sum \limits_{\gamma_t \in E_t\vert_f} a[\gamma_t] \right\vert^2 
\end{equation}
for any $t\ge t_E$. The sum over $t$-path amplitudes in the formula for $\mu(E)$ can be 
re-expressed in terms of $s^{t}_{ifp} (E)$:
\begin{equation}
	\sum \limits_{\gamma_t \in E_t\vert_f} a[\gamma_t] = \frac{1}{n^{t/2}} \sum \limits_{i=0}^{n-1} \psi_i \sum \limits_{p =0}^{n'(n)-1} s^{t}_{ifp} (E) {\omega_n}^p\,.
\end{equation}
$E$ is null if and only if this sum is zero for each final site $f$.

We use the fact that the roots of unity sum to zero
\begin{equation}
	\sum \limits_{p=0}^{n-1} e\left(\frac{p}{n} \right) = 0
\end{equation}
to give the following  identities
\begin{eqnarray}
	\sum \limits_{p=0}^{n-1} {\omega_n}^p = 0  & \qquad \mathrm{for~odd~}n \label{sum_over_roots_of_unity1}
	\\ \sum \limits_{p=0,2}^{2n-2} {\omega_n}^p = \sum \limits_{m=0}^{n-1} e\left(\frac{2m}{2n} \right)  = 0  & \qquad \mathrm{for~even~}n \label{sum_over_roots_of_unity2}
	\\ \sum \limits_{p=1,3}^{2n-1} {\omega_n}^p = e\left(\frac{1}{2n}\right) \sum \limits_{m=0}^{n-1} e\left(\frac{2m}{2n} \right) = 0 & \qquad \mathrm{for~even~} n. \label{sum_over_roots_of_unity3}
\end{eqnarray}
Therefore, for odd $n$, $E$ will be null if, for each $i$ and $f$, the $s^t_{ifp}(E)$ are equal for all $p$. When $n$ is even, $E$ will be null if, for each $i$ and $f$, the $s^t_{ifp}(E)$ are equal when the $p$'s are of the same parity. However, this is trumped by the observation that ${\omega_n}^{p+n} = -{\omega_n}^{p}$ for even $n$. So $E$ will be null if $s^t_{ifp}(E)=s^t_{if(p+n)}(E)$ for all $i,f,p \in \mathbb{Z}_n$.

The events described above will be null no matter what the initial amplitudes are, but there can also be events that are null depending on the initial amplitudes. If $\psi_i = c z_i {\omega_n}^{p_i}$ for some fixed $c \in \mathbb{C}$ and integers $z_i$ and $p_i$ then
\begin{equation}
	\sum \limits_{ \gamma_t \in E_t\vert_f} a[\gamma_t] \propto  \sum \limits_{i=0}^{n-1} \sum \limits_{p=0}^{n'(n)-1} s^{t}_{ifp} (E) z_i {\omega_n}^{p+p_i}
\end{equation}
and we see that it is possible that $ \mu(E) =0$ even when the $s^{t}_{ifp}(E)$ are {not} equal for all $p$. A trivial example of this would be if $\psi_i=0$ for some $i$, then any event that only contains histories that start at $i$ would be null.

Let us look at some examples. 
For the $2$-site hopper, consider the event
\begin{equation}
	A = \cyl{\underbrace{000}_{p=0}} \cup \cyl{\underbrace{010}_{p=2}} \, .
\end{equation}
Then $t_A = 2$, 
\begin{equation}
	s^2_{000}(A)=s^2_{002}(A)=1
\end{equation}
and all other $s^{2}_{ifp}(A) = 0$. So $A$ is null for any initial amplitudes. 

For the $5$-site hopper, the event
\begin{equation}
	 B = \cyl{\underbrace{01203}_{p=0}} \cup  \cyl{\underbrace{00103}_{p=1}} \cup  \cyl{\underbrace{00203}_{p=2}} \cup  \cyl{\underbrace{00123}_{p=3}} \cup  \cyl{\underbrace{00003}_{p=4}}
\end{equation}
has $t_B = 4$ and 
\begin{equation}
	s^4_{03p}(B)=1 \quad \forall ~ p \in \mathbb{Z}_{5}
\end{equation}
with all other $s^{4}_{ifp}(B)  = 0$.  So $B$ is null for any initial amplitudes.  

As a final example, in the $3$-site hopper the event
\begin{equation}
	C = \cyl{\underbrace{000}_{p=0}} \cup  \cyl{\underbrace{010}_{p=2}} \cup \cyl{\underbrace{120}_{p=2}}
\end{equation} 
has $t_C = 2$ and 
\begin{equation}
	s^2_{000}(C)=s^2_{002}(C)=s^2_{102}(C)=1\,,
\end{equation}
with all other $ s^2_{ifp} (C)=0$. 
Let the initial amplitudes satisfy $\psi_0 = \omega_3 \psi_1$. The only final site 
for the  2-paths is $f=0$. So the relevant sum is
\begin{eqnarray}
	\sum \limits_{ \gamma_2 \in C_2\vert_0} a[\gamma_2]& \propto 
	\sum_{ip} \psi_i s^2_{i 0 p} \omega_3^p
	\\ & =\psi_1 s^2 _{102} \omega_3^2 +  \psi_0 s^2 _{000} +  \psi_0 s^2 _{002} \omega_3^2
	\\ & =\psi_1(\omega_3^2 + \omega_3 +1)
	\\ & = 0
\end{eqnarray}
and so $C$ is null. 

\section{Nothing Finite Happens in the Multiplicative Scheme} \label{Nothing_Happens}

\subsection{Multiplicative Co-events}

The set of all events, $\mathfrak{A}$,  for a quantum measure system has a natural Boolean structure inherited  from its being a subset of the power set of $\Omega$. In QMT, possible physical worlds correspond to  ``answering'' maps 
\begin{equation}
	\phi: \mathfrak{A}\rightarrow \mathbb{Z}_2 = \{0,1\}\,.
\end{equation} 
If $\phi(A) = 1$ ($\phi(A) = 0$)  we say that event $A$ is affirmed (denied) by the world $\phi$. Equivalently we can say that $A$ happens (does not happen) in the world $\phi$. 

Such a map, $\phi$, is called a \textit{co-event} \cite{2007_co-events_Rafael_Sorkin,2012_logic_to_quantum_Rafael_Sorkin,2013_coevent_formulation_Petros_Wallden} as a reminder that the map takes events to scalars, one or zero. A co-event is a possible physical world since it gives an answer to every question that can be asked about the system, once the set of histories is fixed. One might say, ``the world is everything that is the case and nothing that is not the case,''  echoing Wittgenstein.

A co-event is considered classical if it preserves the Boolean structure of $\mathfrak{A}$. In particular, it must be the case that the event ``$A$ and $B$'' is affirmed iff $A$ is affirmed and $B$ is affirmed, that ``$A$ xor $B$'' is affirmed iff $A$ is affirmed or $B$ is affirmed but not both, and that ``nothing''$=\emptyset$ is denied and ``something''$=\Omega$ is affirmed. Thus the classical co-events follow our intuitive logic. Moreover, for finite $\Omega$, it can be shown that each classical co-event can be uniquely associated with a single history $\gamma \in \Omega$, where physically $\gamma$ is \textit{the} history that happens. For example, for a ball passing through a left or a right slit, a classical co-event could be associated with the history $\gamma=$`right slit', which would affirm the event ``the ball passes through the right slit'', deny the event ``the ball passes through the left slit'' and affirm the event ``the ball passes through the right or left slit''.

Now, the Law of Preclusion implies that a possible co-event, $\phi$, must satisfy the condition 
\begin{equation}
	 \forall ~ A\in \mathfrak{A}, \quad \mu(A) = 0 ~ \Rightarrow ~ \phi(A) = 0 \,.
\end{equation}
When this condition holds we say that $\phi$ is a \textit{preclusive} co-event. In classic stochastic theory, since $\mu$ would be a probability measure and therefore satisfies the Kolmogorov sum rule, it would always be possible to assign a classical co-event that is preclusive to the system, provided the associated history does not have zero probability. However, with a quantum measure it is possible to cover $\Omega$ with null events, which prevents the assignment of a preclusive classical co-event to the system \cite{2007_co-events_Rafael_Sorkin,2008_kochen_specker_Fay_Dowker_and_Yousef_Ghazi-Tabatabai}.

Instead, in QMT we seek a more general framework for choosing co-events using the law of preclusion. However, this condition is not strong enough to give us all the predictions we want to be able to make. In particular, in the special case that the measure $\mu$ is classical we want to recover the usual classical co-events. To achieve this, more conditions on the possible co-events are required. We refer to a complete set of conditions to be imposed on co-events that defines the full set of possible worlds as a \textit{scheme}. 

In the \textit{multiplicative scheme} \cite{2007_anhomomorphic_logic_Rafael_Sorkin}, the most closely studied scheme to date \cite{2013_3_site_energy_Rafael_Sorkin, 2012_logic_to_quantum_Rafael_Sorkin, 2008_kochen_specker_Fay_Dowker_and_Yousef_Ghazi-Tabatabai}, the possible co-events satisfy
\begin{equation}
	\phi(A\cap B) = \phi(A)\phi(B) \quad \forall ~A, B \in \mathfrak{A}\,.
\end{equation}
The scheme is attractive for several reasons. It can be shown that the multiplicative condition is, essentially, equivalent to the rule of inference known as modus ponens \cite{2012_modus_ponens_Fay_Dowker_and_Kate_Clements_and_Petros_Wallden}. It can also be shown that, in systems with finite $\Omega$, each multiplicative co-event can be uniquely associated with a subset of $\Omega$, know as the \textit{support} of the co-event. This gives the co-events a corresponding partial order, making it natural to add to the scheme the postulate that the allowed co-events are \textit{minimal} in this order. If the measure is classical, then the minimal multiplicative co-events  correspond to single histories in $\Omega$, resulting in classical co-events.  Further, Sorkin has shown that a minimal multiplicative co-event will be classical when restricted to the algebra of macroscopic events with permanent records. So,  within the multiplicative scheme, ``to the extent that one is willing to posit the existence of sufficiently permanent records of macroscopic events, one can [...] regard the measurement problem as solved'' \cite{2012_logic_to_quantum_Rafael_Sorkin}.

Whilst the multiplicative scheme has enjoyed these successes, the previous history spaces it had been applied to were small and had small fixed cut-off times. This is not the case with the $n$-site hopper model presented in this paper, which allows us to expose a problem with the multiplicative scheme. Any future scheme developed for QMT would need to take this problem into account.

\subsection{A ``Problem of Time''}

In a classical theory, if an event $A$ is denied by the physical world (does not happen), then any subevent of $A$ 
is also denied. This property of the physical world also holds in the multiplicative scheme:
\begin{lemma} \label{stymie} 
	Let $\phi$ be a multiplicative co-event and $F$ be an event. If $\phi(F) =0$ and $E\subseteq F$, then $\phi(E) = 0$. 
\end{lemma}
\begin{proof}
	\begin{eqnarray}
		\phi(E) &= \phi(E \cap F)
		\\ &= \phi(E)  \phi(F)
		\\ &= 0\,.
	\end{eqnarray}
\end{proof}
In a classical theory this property is familiar and causes no difficulty. In QMT,  quantum interference results in many more null events than in classical theory and the Law of Preclusion together with the above Lemma then vastly expands the set of events that \textit{must} be denied by \textit{all} allowed co-events. It is useful to have a term for this: if $\mu(A) = 0$, so that $A$ is precluded and $B\subset A$ and $\mu(B) \ne 0$,  we say that in the multiplicative scheme $B$ is \textit{stymied} by $A$.  The fact that in the multiplicative scheme there are events of nonzero measure that must be denied by all allowed co-events turns out to be problematic. 

In the $n$-site hopper, the system is extended arbitrarily in time into the future. The space of $t$-paths, $\Omega_t(n)$, grows exponentially with $t$ and the number of time-finite events grows as an exponential of an exponential. The store of events that could potentially stymie any given fixed time-finite event therefore grows at a huge rate. In particular the most striking result follows from this next theorem, which is a generalisation of a theorem given in \cite{2010_extending_quantum_measure_Fay_Dowker_and_Steven_Johnston_and_Sumati_Surya} for the $2$-site hopper.
\begin{theorem} \label{nirvana} For the $n$-site hopper with any choice of initial amplitudes, for every time-finite event $E$ where $E \not \supseteq \cyl{i} ~ \forall ~i$, there exists a time-finite event $F$ such that  $F \supseteq E$ and $\mu(F)=0$.
\end{theorem}

The consequence of Theorem \ref{nirvana} and Lemma \ref{stymie} for the multiplicative scheme is that all non-trivial 
time-finite events are stymied by a null time-finite event. The only 
time-finite events that are not always denied correspond to statements of the form ``the hopper starts at site $i$ or $\ldots$''. 
In fact, we shall see that for odd $n$ there are certain initial conditions for which the result can be extended to 
stymie these events as well.

We will now prove a number of lemmas before proving Theorem \ref{nirvana} and discussing the implications of this result for multiplicative schemes. Recall that
\begin{equation}
	n'(n):= \cases{
		2n &  for even $n$
		\\ n & for odd $n$.
	} 
\end{equation}
This will allow us to make statements that cover both even and odd $n$.

\begin{lemma} \label{zero_s_2}
	For the $n$-site hopper with even $n$, if a $t$-path with $t\ge 1$ begins at site $i$, ends at site $f$ and has phase ${\omega_n}^p$ where $p \in \mathbb{Z}_{2n}$, then 
	\begin{equation}
		\Emod{p+i+f}{0}{2} \, .
	\end{equation}
\end{lemma}
\begin{proof}
	Consider the $n$-site hopper for any $n$. Consider the $t$-path $\gamma_t \in \Omega_t(n)$ with phase $\tilde{a}[\gamma_t]={\omega_n}^p$ where $p \in \mathbb{Z}_{n'(n)}$. Then 
	\begin{eqnarray}
		\fl p &= \sum \limits_{k=0}^{t - 1} \left[\gamma_t(k) - \gamma_t(k+1)\right]^2 \Modop n'(n)
		\\ \fl &= \sum \limits_{k=0}^{t - 1} \left[\gamma_t(k)^2 + \gamma_t(k+1)^2 - 2\gamma_t(k) \gamma_t(k+1)\right] \Modop n'(n)
		\\ \fl &= \left[\gamma_t(0)^2 + \gamma_t(t)^2  - 2 \gamma_t(0)\gamma_t(1) + 2\sum \limits_{k=1}^{t - 1} \gamma_t(k) \left[ \gamma_t(k) - \gamma_t(k+1) \right]\right]\Modop n'(n)\,. \label{path_phase_expansion}
	\end{eqnarray}

	Now specialise to even $n$, so $n'(n) = 2n$, and fix $\gamma_t(0)=i$ and $\gamma_t(t)=f$. We have
	\begin{equation}
		p = i^2 + f^2 + 2 \left[ \ldots \right]\Modop 2n \, .
	\end{equation}
	Noting that the$\Modop 2n$ operation does not change the parity of the operand, we find
	\begin{equation}
		\Emod{p}{i^2+f^2}{2} 
	\end{equation}
	and, using the fact that the parity of an integer is equal to the parity of its square
	\begin{equation}
		 \Emod{p + i + f}{0}{2}\,.
	\end{equation}
\end{proof}
\begin{corollary}
For the $n$-site hopper with $n$ even, for any time-finite event $E$, if $p \in \mathbb{Z}_{2n}$ and $i, f \in Z_n$ satisfy 
$p + i + f \Modop 2 = 1$,  then $s^t_{ifp}(E) = 0 \quad \forall ~ t \ge t_E$. 
\end{corollary}


Given this result we can restrict the sum over $p$ in the measure calculation for even $n$ to exclude trivially zero $s^t_{ifp}(E)$. Let us define 
\begin{equation}
	\mathbb{Z}_{n,i,f} := \cases{
		\lbrace 0,2,\ldots,2n-2 \rbrace & if $n$ is even and $(i-f)$ is even
		\\ \lbrace 1,3,\ldots,2n-1 \rbrace & if $n$ is even and $(i-f)$ is odd
		\\ \mathbb{Z}_n & if $n$ is odd.
	}
\end{equation}
Then 
\begin{equation}
	\sum \limits_{p=0}^{n'(n)-1} s^t_{ifp}(E) \cdot (\ldots) = \sum \limits_{p \in \mathbb{Z}_{n,i,f}} s^t_{ifp}(E) \cdot (\ldots)\, .
\end{equation}

The next lemma shows that the number of paths in $E_t$, with fixed phase and initial and final sites, grows exponentially with $t$. 
\begin{lemma} \label{lower_s_bound}
	For $t \geq t_E + 2n-1$, $i,f \in \mathbb{Z}_n$ and $p \in \mathbb{Z}_{n,i,f}$, 
	\begin{equation}
		s^{t}_{ifp}(E) \geq c_i(E)n^{t-t_E - 2n + 1} \,,
	\end{equation}
	where $c_i(E)$ is the number of $t_E$-paths in $E_{t_E}$ that start at site $i$.
\end{lemma}

\begin{proof} 
	We will construct a path $\gamma_t \in E_t$, with phase $\tilde{a}[\gamma_t]={\omega_n}^p$, that starts at site $i$ and ends at site $f$, as a concatenation of three paths in three temporal regions as shown in figure \ref{path_growth_diagram}. The first region is from time $0$ to $t_E$, the second from $t_E$ to $t-2(n-1)$, and the third from $t-2(n-1)$ to $t$. Now, for a $t$-path $\gamma_t$ to be in $E_t$ it must match up with one of the $t_E$-paths $\gamma_{t_E} \in E_{t_E}$ for the first $t_E$ steps, but after this there are no restrictions on it. 

	Choose a site $j\in \mathbb{Z}_n$. In the first region, choose  $\gamma^1$ from among the paths in $E_{t_E}$ that start at site $i$ and end at site $j$. Let $c_{ij}(E)$ be the total number of possible choices for $\gamma^1$.
	
	In the second region, we choose  $\gamma^2$ to be any path starting at site $j$ at $t_E$ and ending at $f$ at $t-2(n-1)$. Between these two sites there are $(t-t_E-2n+1)$ time steps, and at each one the path could go to any one of the $n$ sites, so there are $n^{t-t_E-2n+1}$ such paths to choose from.
	
	In the third region $\gamma^3$ starts at site $f$ and ends at $f$, but we must choose it carefully such that  $\gamma_t$, which is the concatenation of $\gamma^1$, $\gamma^2$ and $\gamma^3$, will have a phase of $\tilde{a}[\gamma_t]={\omega_n}^p$. The phase is
	\begin{equation}
		\tilde{a}[\gamma_t]=\tilde{a}[\gamma^1]\tilde{a}[\gamma^2]\tilde{a}[\gamma^3]\,.
	\end{equation}
Since we have already chosen $\gamma^1$ and $\gamma^2$ we have fixed
	\begin{equation}
		\tilde{a}[\gamma^1]\tilde{a}[\gamma^2]={\omega_n}^{q}
	\end{equation}
	for some $q \in \mathbb{Z}_{n,i,f}$, since $\gamma^1$ starts at $i$ and $\gamma^2$ ends at $f$. Therefore, to get a $\gamma_t$ with the right phase we need to choose the path $\gamma^3$ such that its phase is
	\begin{equation}
		\tilde{a}[\gamma^3]={\omega_n}^{(p-q) \Modop n'(n)}\,.
	\end{equation}
	For odd $n$, $(p-q) \Modop n$ could be any number in $\mathbb{Z}_n$. For even $n$, since both $p$ and $q$ are in $\mathbb{Z}_{n,i,f}$, $(p-q) \Modop 2n$ could be any number in $\lbrace 0,2,\ldots,2(n-1) \rbrace$.
	
\begin{figure}[t!]
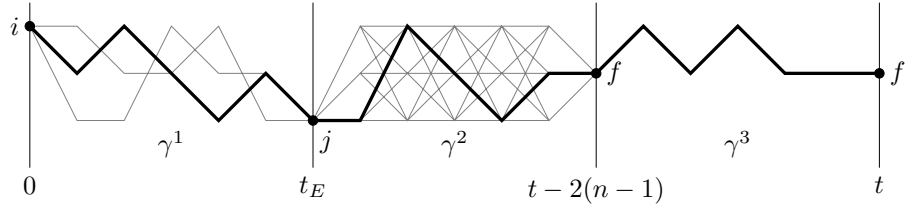

	\newlength{\timelength}
	\setlength{\timelength}{0.04\textwidth}
	\begin{indented}
		\item[] \tikz{
			\draw[line width=0.005\timelength,color=gray] (0,3\timelength) -- ++(\timelength,0) -- ++(\timelength,-\timelength) -- ++(\timelength,0) -- ++(\timelength,\timelength) -- ++(\timelength,-2\timelength) -- ++(\timelength,0)
			(0,3\timelength) -- ++(\timelength,-2\timelength) -- ++(\timelength,0) -- ++(\timelength,2\timelength) -- ++(\timelength,-\timelength) -- ++(\timelength,0) -- ++(\timelength,-\timelength)
			(6\timelength,\timelength) -- ++(\timelength,2\timelength) -- ++(\timelength,-2\timelength) -- ++(\timelength,2\timelength) -- ++(\timelength,-2\timelength) -- ++(\timelength,2\timelength)
			(7\timelength,\timelength) -- ++(\timelength,2\timelength) -- ++(\timelength,-2\timelength) -- ++(\timelength,2\timelength) -- ++(\timelength,-2\timelength)
			(7\timelength,\timelength) -- ++(\timelength,\timelength) -- ++(\timelength,-\timelength)-- ++(\timelength,\timelength) -- ++(\timelength,-\timelength) -- ++(\timelength,\timelength)
			(7\timelength,2\timelength) -- ++(\timelength,\timelength) -- ++(\timelength,-\timelength)-- ++(\timelength,\timelength) -- ++(\timelength,-\timelength)
			(6\timelength,\timelength) -- ++(\timelength,\timelength) -- ++(\timelength,-\timelength)-- ++(\timelength,\timelength) -- ++(\timelength,-\timelength) -- ++(\timelength,\timelength)
			(7\timelength,3\timelength) -- ++(\timelength,-\timelength)-- ++(\timelength,\timelength) -- ++(\timelength,-\timelength)-- ++(\timelength,\timelength) -- ++(\timelength,-\timelength)
			(7\timelength,3\timelength) -- (11\timelength,3\timelength)
			(7\timelength,2\timelength) -- (12\timelength,2\timelength)
			(7\timelength,\timelength) -- (11\timelength,\timelength);
			\draw[line width=0.07\timelength] (0,3\timelength) node[left]{$i$} -- ++(\timelength,-\timelength) -- ++(\timelength,\timelength) -- ++(2\timelength,-2\timelength)  -- ++(\timelength,\timelength) -- ++(\timelength, -\timelength) node[below,align=left]{$~~~j$} -- ++(\timelength,0) -- ++(\timelength,2\timelength) -- ++(2\timelength,-2\timelength) -- ++(\timelength,\timelength) -- ++(\timelength,0) node[right]{$f$} -- ++(\timelength,\timelength) -- ++(\timelength,-\timelength)  -- ++ (\timelength,\timelength) -- ++(\timelength,-\timelength) -- ++ (2\timelength,0) node[right]{$f$};
			\draw[line width=0.02\timelength,postaction={decorate,decoration={markings,mark=at position 0.5\timelength with {\filldraw circle (0.1\timelength)},mark=at position 6\timelength with {\filldraw circle (0.1\timelength)} ,mark=at position 8.5\timelength with {\filldraw circle (0.1\timelength)} ,mark=at position 12\timelength  with {\filldraw circle (0.1\timelength)}}}] (0,3.5\timelength) -- (0,0) node[below]{$0$} (6\timelength,3.5\timelength) -- (6\timelength,0) node[below]{$t_E$} (12\timelength,3.5\timelength) -- (12\timelength,0) node[below]{$t-2(n-1)$} (18\timelength,3.5\timelength) -- (18\timelength,0) node[below]{$t$};
			\node at (3\timelength,0.5\timelength) {$\gamma^1$} ;
			\node at (9\timelength,0.5\timelength) {$\gamma^2$} ;
			\node at (15\timelength,0.5\timelength) {$\gamma^3$} ;
		}
	\end{indented}
	\caption{A diagram of a $t$-path across three temporal regions. The times are written along the bottom, the height of the path specifies the site at that time. Drawn in bold is a specific path in $E_t$, in the first region on the left the path is given by $\gamma^1$ from site $i$ to site $j$, in the middle region it is given by $\gamma^2$ from site $j$ to site $f$, and in the final region it is given by $\gamma^3$ from site $f$ to site $f$. In grey are other choices for $\gamma_1$ and $\gamma_2$ that can be used to find a new path in $E_t$.}
	\label{path_growth_diagram}
\end{figure}

With this in mind, we construct $\gamma^3$ in the following way. It starts by zig-zagging between $f \rightarrow (f+1) \rightarrow f$ a number of times and ends by staying at $f$ for the remaining time steps. The number of time steps in region 3 allows for up to a maximum of $(n-1)$ zig-zags, the exact number of which will determine the amplitude of $\gamma^3$ in the following way
	\begin{eqnarray}
		\mathrm{no~zigzags:} &  \tilde{a}\left[\tikz{\draw (0ex,0ex) -- (6ex,0ex); \draw[dotted] (6ex,0ex) -- (9ex,0ex); \draw (9ex,0ex) -- (13ex,0ex); } \right]={\omega_n}^{0^2+0^2+\ldots+0^2}={\omega_n}^{0}
		\\ \mathrm{one~zigzag:}  & \tilde{a}\left[\tikz{\draw (0ex,0ex) -- (1ex,1ex) -- (2ex,0ex) -- (6ex,0ex); \draw[dotted] (6ex,0ex) -- (9ex,0ex); \draw (9ex,0ex) -- (13ex,0ex); } \right]={\omega_n}^{1^2+1^2+0^2+\ldots+0^2}={\omega_n}^{2}
		\\ \mathrm{two~zigzags:}  & \tilde{a}\left[\tikz{\draw (0ex,0ex) -- (1ex,1ex) -- (2ex,0ex) -- (3ex,1ex) -- (4ex,0ex) -- (6ex,0ex) ; \draw[dotted] (6ex,0ex) -- (9ex,0ex); \draw (9ex,0ex) -- (13ex,0ex); } \right]={\omega_n}^{4}
		\\  & \ldots \nonumber
		\\ n-2 \mathrm{~zigzags:}  & \tilde{a}\left[\tikz{\draw (0ex,0ex) -- (1ex,1ex) -- (2ex,0ex) -- (3ex,1ex) -- (4ex,0ex) -- (5ex,1ex) -- (6ex,0ex) ; \draw[dotted] (6ex,0ex) -- (6.75ex,0.75ex) ; \draw[dotted] (8.25ex,0.75ex) -- (9ex,0ex); \draw (9ex,0ex) -- (10ex,1ex) -- (11ex,0ex) -- (13ex,0ex); } \right]={\omega_n}^{2(n-2) \Modop n'(n)}
		\\  n-1 \mathrm{~zigzags:} \quad  & \tilde{a}\left[\tikz{\draw (0ex,0ex) -- (1ex,1ex) -- (2ex,0ex) -- (3ex,1ex) -- (4ex,0ex) -- (5ex,1ex) -- (6ex,0ex) ;  \draw[dotted] (6ex,0ex) -- (6.75ex,0.75ex) ; \draw[dotted] (8.25ex,0.75ex) -- (9ex,0ex); \draw (9ex,0ex) -- (10ex,1ex) -- (11ex,0ex) -- (12ex,1ex) -- (13ex,0ex); } \right]={\omega_n}^{2(n-1) \Modop n'(n)} \, .
	\end{eqnarray}
	Therefore, choosing such a form for $\gamma^3$, we are able to obtain any phase in
	\begin{equation}
		\left \lbrace {\omega_n}^{2k \Modop n'(n)} ~ \middle\vert ~ k \in \mathbb{Z}_n \right\rbrace = \cases{
			\left\lbrace 1,{\omega_n},\ldots,{\omega_n}^{n-1} \right\rbrace & for odd $n$
			\\ \left\lbrace 1,{\omega_n}^2,\ldots,{\omega_n}^{2(n-1)} \right\rbrace & for even $n$
		}
	\end{equation}
	including the phase ${\omega_n}^{(p-q)\Modop n'(n)}$ required to give $\gamma_t$ the phase we want.

	In conclusion, the number of $\gamma_t$ with phase $\omega_n^p$ that can constructed in this way equals
	\begin{equation}
		\sum \limits_{j=0}^{n-1} c_{ij}(E) n^{t-t_E-2n+1} \, .
	\end{equation}
	Defining $c_i(E)=\sum_j c_{ij}(E)$ we can place the lower bound on $s^t_{ifp}(E)$: the \textit{total} number of $t$-paths in $E_t$ that start at site $i$, end at site $f$, and have phase ${\omega_n}^p$
	\begin{equation}
		s^t_{ifp}(E) \geq c_{i}(E) n^{t-t_E-2n+1} \, .
	\end{equation}
\end{proof}

We will now prove a few lemmas that show that the unitary dynamics is periodic. 
\begin{lemma} \label{U_eigen_vectors_and_values}
	The matrix $U(n)$ for the $n$-site hopper has eigenvectors $[v_j]$, $j \in \mathbb{Z}_n$, with components
	\begin{equation}
		[v_j]_k = \frac{e(jk/n)}{\sqrt{n}}, \quad k \in \mathbb{Z}_n
	\end{equation}
 with corresponding eigenvalues
	\begin{equation}
		\lambda_j =	\cases{
			e\left(\frac{1}{8} \right) e\left(\frac{-j^2}{2n}\right) & for even $n$
			\\ e\left(\frac{-j^2}{4n}\right) & for $\Emod{n}{1}{4}$ and even $j$
			\\ \rmi \, e\left(\frac{-j^2}{4n}\right) & for $\Emod{n}{1}{4}$ and odd $j$
			\\ \rmi \, e\left(\frac{-j^2}{4n}\right) & for $\Emod{n}{3}{4}$ and even $j$
			\\ e\left(\frac{-j^2}{4n}\right) & for $\Emod{n}{3}{4}$ and odd $j$.
		}
	\end{equation}
\end{lemma}
\begin{proof}
	First note that
	\begin{eqnarray}
		U(n)_{jk} &= \frac{{\omega_n}^{(j-k)^2}}{\sqrt{n}}
		\\ &= \frac{{\omega_n}^{((j\pm 1)-(k\pm 1))^2}}{\sqrt{n}}
		\\ &= U(n)_{(j\pm 1)(k \pm 1)}\,.
	\end{eqnarray}
	Therefore $U(n)$ is a circulant matrix. All $n\times n$ circulant matrices have orthonormal eigenvalues $v_j$ of the form \cite{2006_circulant_matrices_Robert_Gray}
	\begin{equation}
		[v_j]_k = \frac{e(jk/n)}{\sqrt{n}} \, ,
	\end{equation}
	where $j,k$ run through $\mathbb{Z}_n$. Now, we will use this to derive the eigenvalues
	\begin{eqnarray}
		\fl \sum \limits_{k=0}^{n-1} U(n)_{k'k} [v_j]_k &= \frac{1}{n} \sum \limits_{k=0}^{n-1} {\omega_n}^{(k-k')^2} e(jk/n)
		\\ \fl &= \frac{1}{n} \sum \limits_{k=-k'}^{n-k'-1} {\omega_n}^{k^2} e(jk/n) e(jk'/n)
		\\ \fl &= \frac{1}{\sqrt{n}} \left[ \sum \limits_{k=n-k'}^{n-1} {\omega_n}^{(k-n)^2} e(j(k-n)/n)  + \sum \limits_{k=0}^{n-k'-1} {\omega_n}^{k^2} e(jk/n) \right] [v_j]_{k'} \, .
	\end{eqnarray}
	But ${\omega_n}^{n^2-2nk}=1$ for both even and odd $n$, and $e(-j)=1$ so
	\begin{equation}
		\sum \limits_{k=0}^{n-1} U(n)_{k'k} [v_j]_k= \frac{1}{\sqrt{n}} \sum \limits_{k=0}^{n-1} {\omega_n}^{k^2} e(jk/n) \, [v_j]_{k'} \, .
	\end{equation}
	This implies	
	\begin{equation}
		\lambda_j = \frac{1}{\sqrt{n}} \cases{
			\sum \limits_{k=0}^{n-1} e\left(\frac{k^2 + 2jk }{2n} \right) & for even $n$
			\\ \sum \limits_{k=0}^{n-1} e\left(\frac{k^2 + jk }{n} \right) & for odd $n$.
		}
	\end{equation}
	The final sums are Gauss sums, which are number theoretic sums of the form
	\begin{equation}
		G(a,b,n) := \sum \limits_{k=0}^{n-1} e\left( \frac{ak^2 +bk}{n} \right) \, ,
	\end{equation}
	where $n$ is a natural number and $a$ and $b$ are usually integers or half-integers. Now, we will use the Landsberg-Schaar relation for Gauss sums \cite{1981_gauss_sums_Bruce_Berndt_and_Ronald_Evans}. For any $a,n \in \mathbb{N}$ and $b \in \mathbb{Z}$, if $an+b$ is even then
	\begin{equation}
		G\left( \frac{a}{2}, \frac{b}{2}, n \right) = \sqrt{\frac{n}{a}} e\left( \frac{1}{8} \right) e\left( \frac{-b^2}{8an} \right) G\left(\frac{n}{2}, \frac{b}{2}, a \right)^* \, .
	\end{equation}
	For a proof to a similar relation that can be easily generalised to this see \cite{2006_number_theory_William_Coppel}.

	Now, for even $n$ we have
	\begin{eqnarray}
		\lambda_j &= \frac{1}{\sqrt{n}} G\left(\frac{1}{2},j,n\right)
		\\ &= e\left( \frac{1}{8} \right) e\left( \frac{-j^2}{2n} \right) G\left(\frac{n}{2}, j, 1 \right)^*
		\\ &= e\left( \frac{1}{8} \right) e\left( \frac{-j^2}{2n} \right) \, .
	\end{eqnarray}
	For odd $n$ we have
	\begin{eqnarray}
		\lambda_j &= \frac{1}{\sqrt{n}} G\left(1,j,n\right)
		\\ &= \frac{1}{\sqrt{2}} e\left( \frac{1}{8} \right) e\left( \frac{-j^2}{4n} \right) G\left(\frac{n}{2}, j, 2 \right)^*
		\\ &=  \frac{1}{\sqrt{2}} e\left( \frac{1}{8} \right)  \left( 1 + e \left(\frac{n}{4} + \frac{j}{2} \right) \right)^* e\left( \frac{-j^2}{4n} \right)
		\\ &= \cases{
			e\left(\frac{-j^2}{4n}\right) & for $\Emod{n}{1}{4}$ and even $j$
			\\ \rmi \, e\left(\frac{-j^2}{4n}\right) & for $\Emod{n}{1}{4}$ and odd $j$
			\\ \rmi \, e\left(\frac{-j^2}{4n}\right) & for $\Emod{n}{3}{4}$ and even $j$
			\\ e\left(\frac{-j^2}{4n}\right) & for $\Emod{n}{3}{4}$ and odd $j$.
		}
	\end{eqnarray}
\end{proof}

\begin{lemma} 
	\textit{}
	\begin{eqnarray}
		U(n)^{2n} = \cases{
			\mathbf{1} & for $\Emod{n}{0}{4}$ 
			\\ -\mathbf{1} & for $\Emod{n}{2}{4}$\,,
		}
		\\ U(n)^n = \cases{ 
			\mathbf{1} & for $\Emod{n}{1}{4}$ 
			\\ -\rmi\mathbf{1} & for $\Emod{n}{3}{4}$\,.
		}
	\end{eqnarray}
\end{lemma}
\begin{proof}
	First, notice that we can use the eigenvectors from Lemma \ref{U_eigen_vectors_and_values} to form unitary matrix $[V]$ with components $[V]_{ij}=[v_i]_j$ that diagonalises $U(n)$. Therefore, if $\exists$ $N\in \mathbb{N}$
	such that ${\lambda_j}^N = \lambda ~ \forall ~ j$,  then $U(n)^N = \lambda \mathbf{1}$.
	
	Now, for $n = 4m$, $m \in \mathbb{N}$
	\begin{eqnarray}
		{\lambda_j}^{2n} &= \left(\frac{1+\rmi}{\sqrt{2}}\right)^{8m} \, e\left(- j^2\right)
		\\ &= 1 \, .
	\end{eqnarray}
	For $n = 4m+2$, $m \in \mathbb{N}$
	\begin{eqnarray}
		{\lambda_j}^{2n} &= \left(\frac{1+\rmi}{\sqrt{2}}\right)^{8m+4} \, e\left(- j^2\right)
		\\ &= -1 \, .
	\end{eqnarray}
	For $n = 4m+1$, $m \in \mathbb{N}$
	\begin{eqnarray}
		{\lambda_j}^{n} &= \cases{
			e\left(\frac{-j^2}{4}\right) & for even $j$
			\\ \rmi^{4m+3} e\left(\frac{- j^2}{4}\right) & for odd $j$
		}
		\\ &= 1 \, .
	\end{eqnarray}
	For $n = 4m+3$, $m \in \mathbb{N}$
	\begin{eqnarray}
		{\lambda_j}^{n} &= \cases{
			\rmi^{4m+3} \, e\left(\frac{-j^2}{4}\right) & for even $j$
			\\ e\left(\frac{-j^2}{4}\right) & for odd $j$
		}
		\\ &= -\rmi \, .
	\end{eqnarray}
\end{proof}

\begin{corollary} \label{U_unity}
	For all $n$ site hoppers
	\begin{equation}
		U(n)^{4n} = \mathbf{1}\,.
	\end{equation}
\end{corollary}

We will now prove Theorem \ref{nirvana}. 
\begin{proof}[Proof of Theorem \ref{nirvana}]
	Let $E$ be a time-finite event for the $n$-site hopper such that $E \not \supseteq \cyl{i} ~ \forall ~i$. Let $ {\bar{E}}: = \Omega (n)\setminus E$ be the complement of $E$. Note that $\bar{E}$ is a time-finite event and $t_{\bar{E}} = t_E$. We will construct a superset  of $E$, $F=E \cup G$, where $G \subseteq {\bar{E}}$, such that $\mu(F)=0$. We will do this by choosing a $t > t_E$ large enough and selecting $t$-paths from $\bar{E}_t$ with amplitudes that cancel the amplitudes in $\mu(E)$. $G$ is then the union of the cylinder sets of these selected $t$-paths. 

	We have $E \cap G = \emptyset$ so for any $t>t_E$
	\begin{equation}
		\mu(E \cup G) =  \sum \limits_{f=0}^{n-1} \left\vert \sum \limits_{ \gamma_t \in E_t\vert_f} a[\gamma_t] +  \sum \limits_{ \gamma_t \in G_t\vert_f} a[\gamma_t]\right\vert^2\,. \label{measure_E_cup_G}
	\end{equation}
	Notice that
	\begin{eqnarray}
		\sum \limits_{ \gamma_{t} \in E_t\vert_f} a[\gamma_{t}] &= \sum \limits_{ \gamma_{t-1} \in E_{t-1}} a[\gamma_{t-1}f]
		\\ &= \sum \limits_{j=0}^{n-1} \sum \limits_{ \gamma_{t-1} \in E_{t-1} \vert_j} a[\gamma_{t-1}]U(n)_{jf}\,, \label{sum_over_amplitudes_under_time_transform}
	\end{eqnarray}
	where we used that $E_{t-1}$ still contains all the histories needed to reproduce $E$ because $(t-1) \geq t_E$. Now, applying (\ref{sum_over_amplitudes_under_time_transform}) iteratively to (\ref{measure_E_cup_G}) gives
	\begin{eqnarray}
		\mu(E \cup G)= \sum \limits_{f=0}^{n-1} \left\vert  \sum \limits_{j=0}^{n-1} \left[U(n)^{t-t_E}\right]_{jf} \sum \limits_{ \gamma_{t_E} \in E_{t_E}\vert_j} a[\gamma_{t_E}]  +  \sum \limits_{ \gamma_t \in G_t\vert_f} a[\gamma_t]\right\vert^2\,.
	\end{eqnarray}
	Then we choose $t = t_E + 4nm$ for some $m \in \mathbb{N}$, so we can use corollary \ref{U_unity} to replace $U(n)^{t-t_E}$ with the identity:
	\begin{eqnarray}
		\mu(E \cup G)&= \sum \limits_{f=0}^{n-1} \left\vert \sum \limits_{\gamma_{t_E} \in E_{t_E}\vert_f} a[\gamma_{t_E}] +  \sum \limits_{ \gamma_t \in G_t\vert_f} a[\gamma_t]\right\vert^2
		\\ &= \frac{1}{n^{t_E}} \sum \limits_{f=0}^{n-1} \left\vert \sum \limits_{i=0}^{n-1} \psi_i \sum \limits_{p \in \mathbb{Z}_{n,i,f}} \left[ s^{t_E}_{ifp}(E) + \frac{s^t_{ifp}(G)}{n^{(t-t_E)/2}} \right]{\omega_n}^p \right\vert^2\,. \label{measureG}
	\end{eqnarray}
	
	The number of paths in $E_t$ that start at site $i$, end at site $f$ and have phase ${\omega_n}^p$ is bounded from below by Lemma \ref{lower_s_bound}
	\begin{equation}
		s^t_{ifp}(\bar{E}) \geq  c_i(\bar{E}) n^{t-t_E-2n+1} \,.
	\end{equation}
	Since $E \not \supseteq \cyl{i}$ for any initial site $i$, $\bar{E}$ must contain paths that starts at site $i$, so $c_i(\bar{E}) \geq 1$. Therefore, for every $i,f \in \mathbb{Z}_n$ and $p \in \mathbb{Z}_{n,i,f}$, $s^t_{ifp}(\bar{E})$ will scale at least as fast $n^t$. In particular, this means there will exist a finite $m$, and therefore a finite $t=t_E+4nm$, such that
	\begin{equation}
		s^t_{ifp}(\bar{E}) \geq  n^{(t-t_E)/2} \left( \max \limits_{q\in \mathbb{Z}_{n,i,f}} s^{t_E}_{ifq}(E) - s^{t_E}_{ifp}(E) \right) \,.
	\end{equation}
	Given such a $t$, since $G$ can be any subset of $\bar{E}$, we can choose unique paths in $\bar{E}_t$ that start at site $i$ and end at site $f$ with phase ${\omega_n}^p$ to construct a $G_t$ such that
	\begin{equation}
		s^t_{ifp}(G) = n^{(t-t_E)/2} \left( \max \limits_{q\in \mathbb{Z}_{n,i,f}} s^{t_E}_{ifq}(E) - s^{t_E}_{ifp}(E) \right) 
	\end{equation}
	for all $i,f \in \mathbb{Z}_n$ and $p \in \mathbb{Z}_{n,i,f}$.
	Note that the right hand side is a non-negative integer so this is allowed.
	Substituting $s^t_{ifp}(G)$ into (\ref{measureG}) gives
	\begin{eqnarray}
		 \mu(E \cup G)
		&= \frac{1}{n^{t_E}} \sum \limits_{f=0}^{n-1} \left\vert \sum \limits_{i=0}^{n-1} \psi_i \, \max \limits_{q\in\mathbb{Z}_{n,i,f}} s^{t_E}_{ifq}(E) {\sum \limits_{p \in \mathbb{Z}_{n,i,f} } {\omega_n}^p } \right\vert^2
		\\ &= 0 \, ,  
	\end{eqnarray}
	where we used (\ref{sum_over_roots_of_unity1}), (\ref{sum_over_roots_of_unity2}) and (\ref{sum_over_roots_of_unity3}) in the last step.
\end{proof}
	
\subsection{Null events dependent on the initial amplitudes}

We will now show that, for certain initial amplitudes, it is even possible to stymie 
``initial position'' events, i.e. events $E$ where $E \supseteq \cyl{i}$ for some $i$ 
for which $\psi_i \ne 0$. 

 We first define the following set of initial sites
\begin{equation}
	I(E) :=\left\lbrace i ~ \middle\vert ~ E \supseteq \cyl{i} \right\rbrace\,.
\end{equation}
For this scenario we will split our event $E$ into the subsets $E^{(i)}$ defined by
\begin{equation}
	E^{(i)} := \left\lbrace \gamma \in E ~ \middle\vert ~ \gamma(0)=i \right\rbrace\,.
\end{equation}
Note that for $i \in I(E)$, $E^{(i)}=\cyl{i}$, and thus $t_{E^{(i)}}=0$. Therefore for all $t \geq t_E$
\begin{eqnarray}
	\sum \limits_{\gamma_t \in E_t\vert_f} a[\gamma_t] &= \sum \limits_{i \in I(E)} \sum \limits_{\gamma_t \in {E^{(i)}_t\vert_f}} a[\gamma_t] \, + \, \sum \limits_{i \not \in I(E)} \sum \limits_{\gamma_t \in {E^{(i)}_t\vert_f}} a[\gamma_t]
	\\ &= \sum \limits_{i \in I(E)} \sum \limits_{j=0}^{n-1} \left[ U(n)^{t} \right]_{jf} \sum \limits_{\gamma_0 \in {E^{(i)}_0\vert_f}} a[ \gamma_0] \, + \, \sum \limits_{i \not \in I(E)} \sum \limits_{\gamma_t \in {E^{(i)}_t\vert_f}} a[\gamma_t]\,,
	\end{eqnarray}
	by applying (\ref{sum_over_amplitudes_under_time_transform}) iteratively.
	Now, ${E^{(i)}_0} = \lbrace \gamma_0 \rbrace$, where $\gamma_0(0)=i$ and so
\begin{equation}	
	\sum \limits_{\gamma_t \in E_t\vert_f} a[\gamma_t] = \sum \limits_{i \in I(E)} \sum \limits_{j=0}^{n-1} \left[ U(n)^{t} \right]_{jf} \psi_i \delta_{ij} \, + \, \sum \limits_{i \not \in I(E)} \sum \limits_{\gamma_t \in {E^{(i)}_t\vert_f}} a[\gamma_t]\,.
\end{equation}
And for further simplicity we can choose a $t=4nm \geq t_E$ for some large enough $m \in \mathbb{N}$ to make $U(n)^t=\mathbf{1}$ via Corollary \ref{U_unity}. Now, we want to find a null superset $F \supset E$. Following the same method as before we split it into two parts $F=E \cup G$, $E \cap G = \emptyset$. Note that necessarily $G^{(i)}=\emptyset$ for $i\in I(E)$, therefore
\begin{equation}
	\mu(F)=\sum \limits_{f=0}^{n-1} \Bigg\vert \sum \limits_{i \in I(E)} \psi_i \delta_{if} \, + \, \sum \limits_{i \not\in I(E)} \sum \limits_{\gamma_t \in {F^{(i)}_t\vert_f}} a[\gamma_t]
 \Bigg\vert^2\,.
\end{equation}

We will now restrict ourselves to an $n$-site hopper with odd $n$. Furthermore, we will consider the `worst case scenario' of an event $E$ which does \textit{not} contain \textit{all} paths that start at site $\check{\imath}$, but does contain all paths that do \textit{not} start at site $\check{\imath}$, such that $I(E)=\mathbb{Z}_n \setminus \lbrace \check{\imath} \rbrace$. Note that if we can find a null superset for all such events then this implies we can do so for every time-finite event $E \neq \Omega(n)$. In addition, we will choose initial amplitudes of the form $\psi_i = c z_i {\omega_n}^{q_i}$ for some (non-unique) $c \in \mathbb{C}$, $z_i \in \mathbb{Z}$ and $q_i \in \mathbb{Z}_n$. Note that we could have generalised $z_i$ to be a rational number, but then we could just write all the $z_i$ with the same denominator and absorbed it into $c$. Without loss of generality we will choose $q_{\check{\imath}}=0$ by changing $c$.

Whilst maintaining $t=4nm$, we will also choose $T=4nM$, $M>m$, such that $U(n)^{T-t}=\mathbf{1}$. Then, as in Theorem \ref{nirvana}'s proof, we can split the sum over path amplitudes for $F_t$ into two disjoint sums over $E_t$ and $G_T$ to eventually derive
\begin{eqnarray}
	\mu(F) &= \sum \limits_{f=0}^{n-1} \left\vert \psi_f (1-\delta_{\check{\imath}f}) + \psi_{\check{\imath}} \sum \limits_{p =0}^{n-1} \left[ \frac{s^t_{\check{\imath}fp}(E)}{n^{t/2}} + \frac{s^T_{\check{\imath}fp}(G)}{n^{T/2}} \right] {\omega_n}^p \right\vert^2
	\\ & \propto \sum \limits_{f=0}^{n-1} \left\vert \sum \limits_{p=0}^{n-1} \left[ \delta_{pq_f}(1-\delta_{\check{\imath}f}) z_f + z_{\check{\imath}} \frac{s^t_{\check{\imath}fp}(E)}{n^{t/2}} +  z_{\check{\imath}} \frac{s^T_{\check{\imath}fp}(G)}{n^{T/2}} \right] {\omega_n}^p \right\vert^2 \, .
\end{eqnarray}
If $z_{\check{\imath}}=0$ then we cannot find a $G$ such that $\mu(F)=0$, which makes sense since all the histories in $G$ would have zero amplitudes. We will now assume $z_{\check{\imath}}\neq 0$.

Now, as before, we must form $G$ using the cylinder sets of paths in $\bar{E}_T$, and again we want to choose a particular number of paths that start at $\check{\imath}$, end at site $f$ and with a phase of ${\omega_n}^p$, but this is restricted by the number $s^T_{\check{\imath}fp}(\bar{E})$. However we know from Lemma \ref{lower_s_bound} that this is bounded from below by
\begin{equation}
	s^T_{\check{\imath}fp}(\bar{E}) \geq c_{\check{\imath}} n^{T-t_E-2n+1} \, .
\end{equation}
Naively, one might use the arguments from Theorem \ref{nirvana}'s proof to also argue that, provided we choose a large enough $T$, we should be able to choose a $G$ such that
\begin{equation}
	\fl s^T_{\check{\imath}fp}(G) = n^{(T-t)/2} \left( \max \limits_{q} s^t_{\check{\imath}fq}(E) - s^t_{\check{\imath}fp}(E) \right)  + n^{T/2}(1 - \delta_{\check{\imath}f}) \left( (1-\delta_{pq_f}) \frac{z_f}{z_{\check{\imath}}} + \frac{\vert z_f \vert}{\vert z_{\check{\imath}} \vert} \right)
\end{equation}
for all $f,p \in \mathbb{Z}_n$. However, such a choice would not be possible if the right hand side is not a non-negative integer. The right hand side is non-negative but is not necessarily an integer, and so we find that $z_{\check{\imath}}$ must divide $n^{T/2}z_f$ for all $f$ for some $T=4nM$. Note that if we satisfy this at some $M$ then it remains satisfied at $M'>M$ and so we are still free to make $M$ larger.

Therefore, if $z_{\check{\imath}}$ divides $n^{T/2}z_f$, we can make this choice for $G$, and we can use (\ref{sum_over_roots_of_unity1}) to show that $\mu(F)=0$. Note that if we tried to apply this argument to the even $n$-site hopper we would have the further restriction that $s^T_{\check{\imath}fp}(G)=0$ necessarily if $p \not \in \mathbb{Z}_{n,\check{\imath}-f}$, and so if $q_f \not \in \mathbb{Z}_{n,\check{\imath}-f}$ we would not be able to satisfy the condition.

\section{Discussion}\label{analysis}

Our result shows that the Multiplicative Scheme for the $n$-site hopper has a rather severe ``infrared'' problem stemming from the existence of null events with arbitrarily large defining times. To what extent does this indicate a problem with the Multiplicative Scheme more generally? 

One possibility for escaping the conclusion that no time-finite event happens is to postulate that the universe is finite in extent in time. For the $n$-site hopper, this  translates into a time $t_{\mathrm{max}}$ at which the universe ends, so we are not free to look for stymieing null events with defining times larger than $t_{\mathrm{max}}$. Inspecting the proof of  Theorem \ref{nirvana} we see that to stymie an event $E$, we constructed a null event with defining time that was at least of order $n$ larger than $t_E$. So, for example, if $n\gg t_{\mathrm{max}}$, then it 
is likely that this issue can be avoided. 

Even if the ``no time-finite event happens'' problem can be alleviated in some way, the stymieing of events by 
null events with later defining times is inherent in the Multiplicative Scheme. Consider for example the $2$-site hopper event,
\begin{eqnarray}
	E &= \cyl{001} \cup \cyl{0000} \cup \cyl{0101}
	\\ & =  \cyl{0010}\cup \cyl{0011} \cup \cyl{0000} \cup \cyl{0101}
\end{eqnarray}
with $t_E=3$. Then we have,  $s^3_{000}=s^3_{002}=s^3_{011}=s^3_{013}=1$ 
which implies that $E$  is null.  $E$ stymies subevents with the same defining time, such as  $\cyl{0000} \cup \cyl{0101}$. It also stymies events of the form ``$E$ \textit{and then} something''. An example is $\cyl{00101} \cup \cyl{00111} \cup{00001} \cup \cyl{01011}$ which corresponds to ``$E$ and then the hopper is at site $1$ at time $t=4$''. These are examples of physical inferences of  a sort familiar from classical physics. But stymieing also blocks events to $E$'s past, such as $\cyl{001}$. It is this latter effect, this essentially  global-in-time nature of the multiplicative scheme which is  disturbing, and it is not confined to the $n$-site hopper.\footnote{These problems are absent when the theory is classical  because a classical (probability) measure obeys $\mathrm{Prob}(G)=0 \Rightarrow \mathrm{Prob}(E)=0$ for any $G \supset E$. Therefore, there is no stymieing since subevents of a null event are null and precluded in their own right.}

This could be an appealing result to anyone who thinks of the universe as a single, simultaneously existing, spacetime \textit{block} \cite{2002_block_view_Paul_Davies,2011_no_flow_time_Huw_Price}, but for those who favour the concept of Becoming, the stymieing of events in the past is in conflict with the idea of a physical passage of time  \cite{2010_causal_sets_non_existent_future_Rafael_Sorkin,2014_causal_sets_birth_Fay_Dowker}. Moreover, this feature makes physical prediction highly impractical: to know what events are possible at early times we have to look at all future events and their measures first.

One could argue that the $n$-site hopper is an exception because it is highly symmetric and regular, which means that the number of $t$-paths in $\Omega(n)_t$ grows exponentially whilst the number of $t$-path phases is fixed at $ n'(n)$, resulting in many null events. However, it seems plausible that there exists a larger class of models with unitary transfer matrices, $U$, such that for every time-finite event $E$ there exists a sequence of time-finite supersets $F_k \supset E$ with increasing $t_{F_k}$, such that for any given $\epsilon>0$ we can find a $K$ such that $\mu(F_k) < \epsilon \quad \forall ~ k \geq K$. The question of how to deal with ``approximate preclusion'' is an open one in Quantum Measure Theory but for such a situation, it would be hard to avoid the conclusion that $E$ would be stymied.

To work towards finding such a class of models, consider a hopper on the same $n$-spatial-sites lattice, with the same set of histories $\Omega(n)$ but with a general unitary transfer matrix $U$ of dimension $n$. 
And consider an event $E$.  Then for $t > t_E$
\begin{eqnarray}
	\left\vert \sum \limits_{\gamma_{t} \in E_t\vert_f} a[\gamma_{t}] \right\vert &= \left\vert \sum \limits_{j=0}^{n-1} [U^{t-t_E}]_{jf} \sum \limits_{\gamma_{t_E} \in E_{t_E}\vert_j} a[\gamma_{t_E}] \right\vert
	\\ & \leq \sum \limits_{j=0}^{n-1} \left\vert [U^{t-t_E}]_{jf} \right\vert \left\vert \sum \limits_{\gamma_{t_E} \in E_{t_E}\vert_j} a[\gamma_{t_E}] \right\vert
	\\ & \leq \sum \limits_{j=0}^{n-1} \left\vert \sum \limits_{\gamma_{t_E} \in E_{t_E}\vert_j} a[\gamma_{t_E}] \right\vert \, ,
\end{eqnarray}
where we used that $\vert U_{jk} \vert \leq 1$ for a unitary matrix. So whilst in the $n$-site hopper we could find times where the sum of amplitudes over $t$-paths in $E_t$ was unchanged, for a more general unitary system this sum of amplitudes is bounded in magnitude from above. And so one could potentially proceed along the same lines as the proof for Theorem \ref{nirvana} by arguing that for a large enough $t$ there are enough paths outside of $E$ with the right amplitudes to cancel $E$'s amplitude, and since the amplitudes are becoming smaller in magnitude for longer paths, but not too small, we can make this cancellation more and more precise to make the measure as small as we like. There is reason to believe that one could show this using the fact that the quantum measure for most unitary systems is not of bounded variation \cite{2010_extending_quantum_measure_Fay_Dowker_and_Steven_Johnston_and_Sumati_Surya}.

There is another, related, aspect of the multiplicative scheme that seems problematic: for unitary systems the only way that a history, $\gamma\in \Omega$, can distinguish itself \textit{dynamically} from another, $\gamma'$ that ends at the same position, is via its amplitude. If the amplitudes of $\gamma$  and $\gamma'$ are equal the dynamical laws that govern the allowed co-events of the Multiplicative Scheme treat them the same. For example, given a null event, $E$, we can replace any history, $\gamma \in E$ with another history, $\gamma' \notin E$ with the same amplitude to get a new null event, $E'$. And if $\gamma$ is an element of the support of an allowed multiplicative co-event, then it can be replaced by $\gamma'$ with the same amplitude to give another allowed co-event. If the $n$-site hopper is a guide, there will always be many more histories than there are amplitudes and there will be wildly different histories sharing the same amplitude. It seems that \textit{any}  scheme that treats all histories with the same amplitude equally would struggle to provide an explanation of the emergence of classical behaviour from quantum theory. 

There are at least two possible avenues for revising the scheme: change the preclusion law or give up multiplicativity.  As a physical guide to how to make these revisions we can look to the nature of the model as a temporal process and to the heuristic of Becoming. For example, we could consider relaxing the strict Law of Preclusion and instead adopt a new rule that null events that would stymie earlier events are not precluded: ``Prior Existence trumps Law.'' Or, we could keep the strict preclusion rule but drop the multiplicative condition in favour of an \textit{evolving} construction in which, for each $t$ there is a \textit{time-finite co-event} $\phi_t$ on the subalgebra, $\mathfrak{A}_t$ of events with defining time $t$ or less. This partial co-event corresponds to the physical world of the hopper system up to time $t$. We build up the co-event step by step and require that at each time $t$ the partial co-event agrees with the previous partial co-event on subalgebra $\mathfrak{A}_{t-1}$. Note that the number of potential partial co-events  grows as $2^{2^{\vert \Omega_t \vert}}$ and there must be additional conditions on allowed co-events in order for the scheme to be viable. This is under current investigation. 

Finally, the $n$-site hopper was introduced as a discrete analogue of the free  particle in the continuum. We can ask to what extent this is more than an analogy: is there a continuum limit for the $n$-site hopper which gives the continuum propagator? Consider the $n$-site hopper propagator, $\left[U(n)^t \right]_{if}$. Using the eigenvectors and values from Lemma \ref{U_eigen_vectors_and_values}, we have
\begin{equation}
	\fl \left[U(n)^t \right]_{if} = \cases{
		\frac{1}{2n} e\left( \frac{t}{8} \right) \sum \limits_{k=0}^{2n-1} e \left( \frac{2(i-f)k -tk^2 }{2n} \right) & for even $n$
		\\ \frac{1}{n} \sum \limits_{k=0}^{n-1} e \left( \frac{2(i-f)k -tk^2}{n} \right) & for $\Emod{n}{1}{4}$ 
		\\  \frac{1}{n} e\left( \frac{t}{4} \right) \sum \limits_{k=0}^{n-1} e \left( \frac{2(i-f)k-tk^2}{n} \right) & for $\Emod{n}{3}{4}$\,.
	}
\end{equation}
These  Gauss sums  can be transformed to other Gauss sums over $\mathcal{O}(t)$ terms using the Landsberg-Schaar relation given in Lemma \ref{U_eigen_vectors_and_values}. The values of these sums are not generally known, however we can compare this to the propagator for a particle on a ring of radius $R$
\begin{equation}
	K^{(\mathrm{ring})}(\phi_f,t^{(c)};\phi_i,0) = \frac{1}{2\pi} \sum \limits_{k=-\infty}^{\infty} \mathrm{exp}\left( \rmi (\phi_f - \phi_i) k - \rmi \frac{\hbar}{2mR^2} t^{(c)} k^2 \right) \, ,
\end{equation}
where $\phi$ is the angular position on the ring. This \textit{looks} very similar in form to the Gauss sums for the $n$-site hopper. It remains to be shown if there is a way to take the limit $n\rightarrow \infty$ so that $\left[U(n)^t \right]_{if}$ (appropriately rescaled) tends to the propagator on the ring, or the free propagator on the line. If so, this would increase the relevance of the $n$-site hopper. The decoherence functional or quantum measure in the continuum is not yet able to be defined for events other than finite unions of events corresponding to finite sequences of projections onto 
position \cite{1993_spacetime_QM_lectures_Jame_Hartle, 2012_towards_a_fundamental_Rafael_Sorkin, 2010_history_hilbert_construction_Fay_Dowker_Steven_Johnston_Rafael_Sorkin, 2013_path_integral_pitfall_Jonathan_Halliwell_and_James_Yearsley}. If the $n$-site hoppers have a continuum limit  one could potentially use this to \textit{define} the quantum measure of an event in the continuum as the limit of a sequence of appropriate $n$-site hopper events as $n\rightarrow \infty$.

\ack
We thank Rafael Sorkin for helpful discussions and the suggestion to use Gauss sums 
to calculate powers of the transfer matrix. Research at Perimeter Institute for Theoretical Physics is supported in part by the Government of Canada through NSERC and by the Province of Ontario through MRI.
FD is supported by STFC grant ST/L00044X/1. 
HW is supported by the STFC.


\end{document}